\documentclass[a4paper,10pt]{article}

\usepackage{amsmath, amsthm, amssymb}
\usepackage{graphicx}

\newtheorem{theorem}{Theorem}[section]

\newtheorem{proposition}[theorem]{Proposition}

\begin{document}
\title{Interlaced particle systems and tilings of the Aztec diamond}
\author{Benjamin J. Fleming and Peter J. Forrester}
\date{}
\maketitle

\noindent
\thanks{\small
Department of Mathematics and Statistics,
The University of Melbourne,
Victoria 3010, Australia \\
email: {B.Fleming2@pgrad.unimelb.edu.au; P.Forrester@ms.unimelb.edu.au}

\begin{abstract}
\noindent
Motivated by the problem of domino tilings of the Aztec diamond, a weighted particle system is defined on $N$ lines, with line $j$ containing $j$ particles. The particles are restricted to lattice points from 0 to $N$, and particles on successive lines are subject to an interlacing constraint. It is shown that marginal distributions for this particle system can be computed exactly. This in turn is used to give
unified derivations of a number of fundamental properties of the tiling problem, for example the evaluation of the number of distinct configurations and the relation to the GUE minor process. An interlaced particle system associated with the domino tiling of a certain half Aztec diamond
is similarly defined and analyzed.
\end{abstract}

\section{Introduction}

The analysis of certain tiling models
is of common interest to both combinatorics and statistical mechanics.  As an explicit example, consider a so-called $(a,b,c)$ hexagon.  This is a hexagon with integer side lengths $a,b,c,a,b,c$ --- side lengths $a$ vertical by convention --- reading anti-clockwise, and all internal angles $\frac{2\pi}{3}$ (see Figure \ref{Hexagon} for an example).  Such a hexagon can be tiled using three species of rhombi, each with side lengths $1$ and angles $\frac{\pi}{3}$, $\frac{2\pi}{3}$.  The three species of rhombi are distinguished by their orientation -- down sloping, up sloping or neutral in slope, reading left to right.  As illustrated in Figure \ref{Hexagon}, it is immediately clear that a particular tiling of the hexagon can be uniquely specified by a family of non-intersecting lattice paths.  These all start and finish one unit apart, and move up or down half a unit at each step (reading left to right).

\begin{figure}[p]\label{Hexagon}
\includegraphics[scale=0.6]{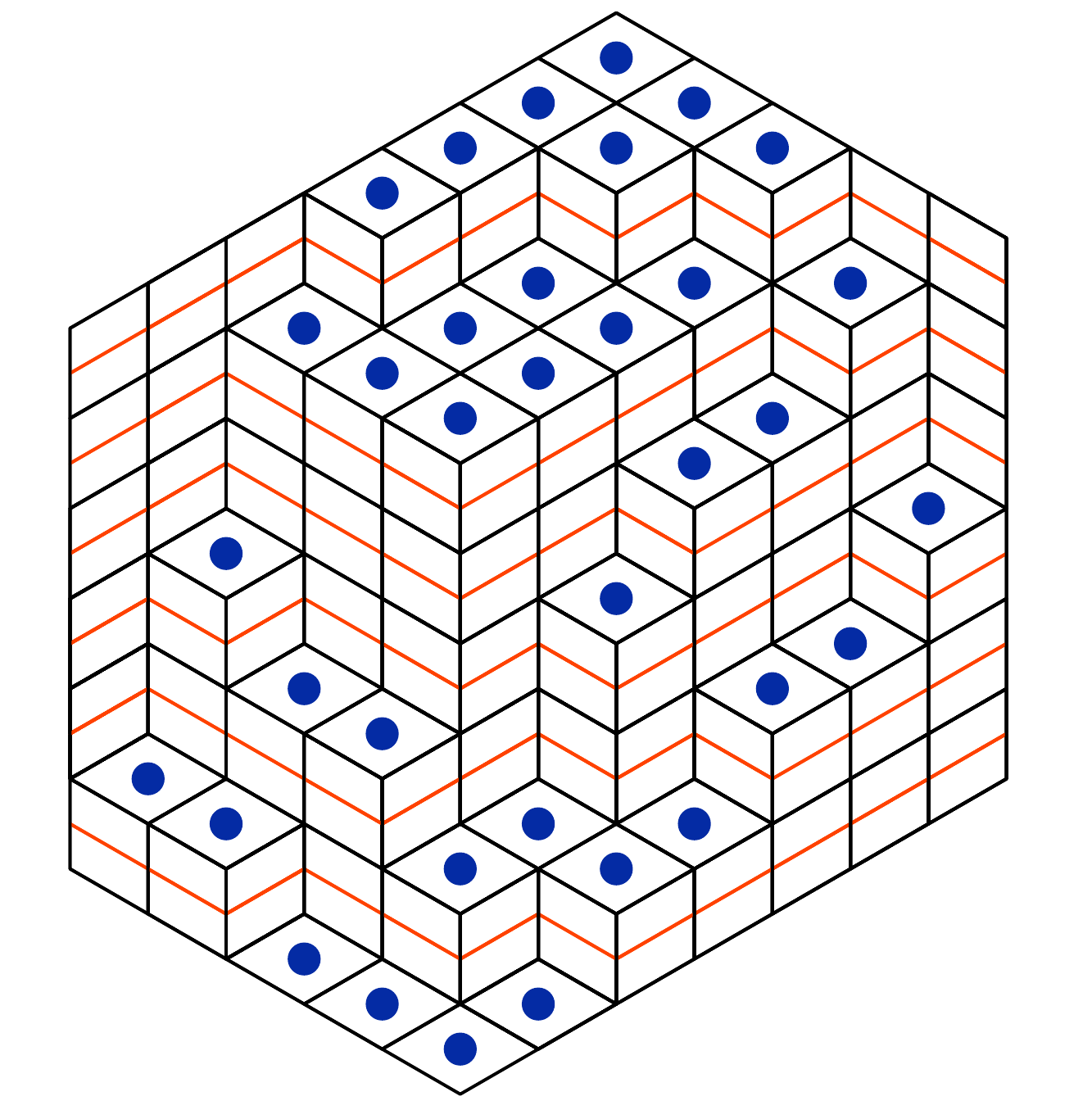}
\caption{(Colour online) A (6,5,7) hexagon, showing the family of non-intersecting lattice paths
(red lines), as well as the particles in the centres of the horizontal rhombi.}
\end{figure}

\begin{figure}[p]\label{Aztec Paths}
\includegraphics[scale=0.65]{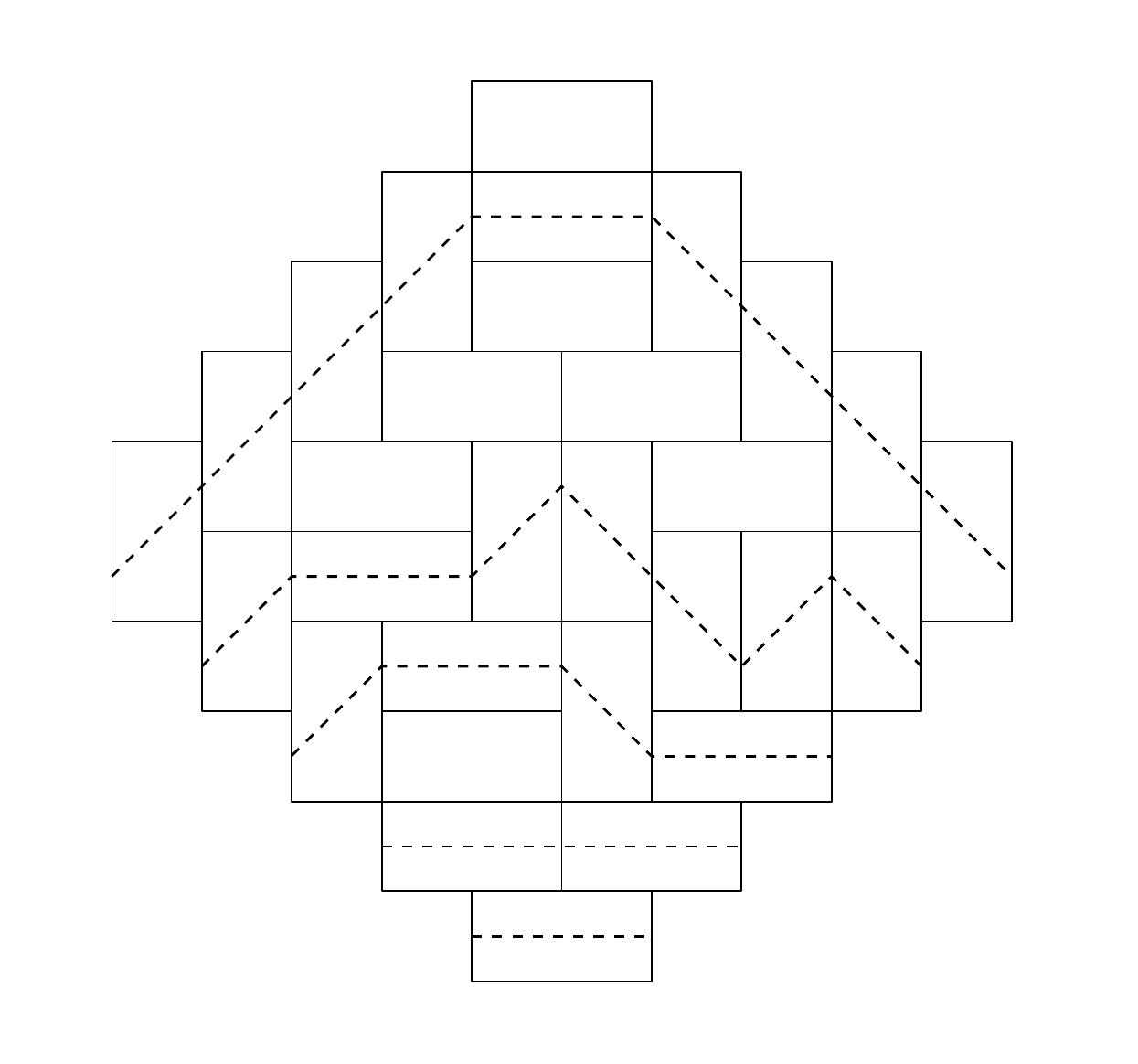}
\caption{The 5 non-intersecting paths corresponding to a particular domino tiling
of the Aztec diamond of order 5.}
\end{figure}

In this article our interest is in tilings of the Aztec diamond by $2\times1$ dominoes. The Aztec diamond of order $N$ is the union of all lattice squares within the diamond shaped region $\{ (x,y): |x| +|y| \leq N+1\}$, and the dominoes may cover the lattice squares by being placed horizontally or vertically.  As with the hexagon tiling of the previous paragraph, such a tiling of the Aztec diamond can be uniquely specified by a family of lattice paths.  To see this, with the top left lattice square specified as white, introduce a checkerboard colouring of all the lattice squares making up the Aztec diamond.  For a horizontal domino which covers a white-black (black-white) pair of squares when reading left to right, no segment (a horizontal segment) of path is marked.  For a vertical domino which covers a white-black (black-white) pair of squares when reading top to bottom, a right-up (right-down) segment of path is marked.  This results in a family of $N$ non-intersecting lattice paths, with segments up sloping, down sloping or horizontal, starting at equally spaced points on the bottom down sloping edge, and finishing at the corresponding points on the bottom up sloping edge.  See Figure 2 for an example.

In the case of the tiling of the hexagon, Figure \ref{Hexagon} makes it clear that complementary to the non-intersecting paths, the tiling configuration can equally as well be specified by recording only the centres of the neutral in slope rhombi.  These centres can in turn be regarded as a particle system on the set of vertical lines naturally associated with the $(a,b,c)$ hexagon.  This particle system has the peculiar property that the number of particles equals the number of lines for lines $1,2,\dots,b$, then stays constant for lines $b+1,\dots,c$, and then equals $b-1,b-2,\dots,1$ for lines $c+1, \dots, c+b-1$.  Furthermore, the particles are constrained by interlacing constraints, the details of which are evident by inspection of Figure \ref{Hexagon}.

\begin{figure}[t]\label{Aztec Shading}
\includegraphics[scale=0.7]{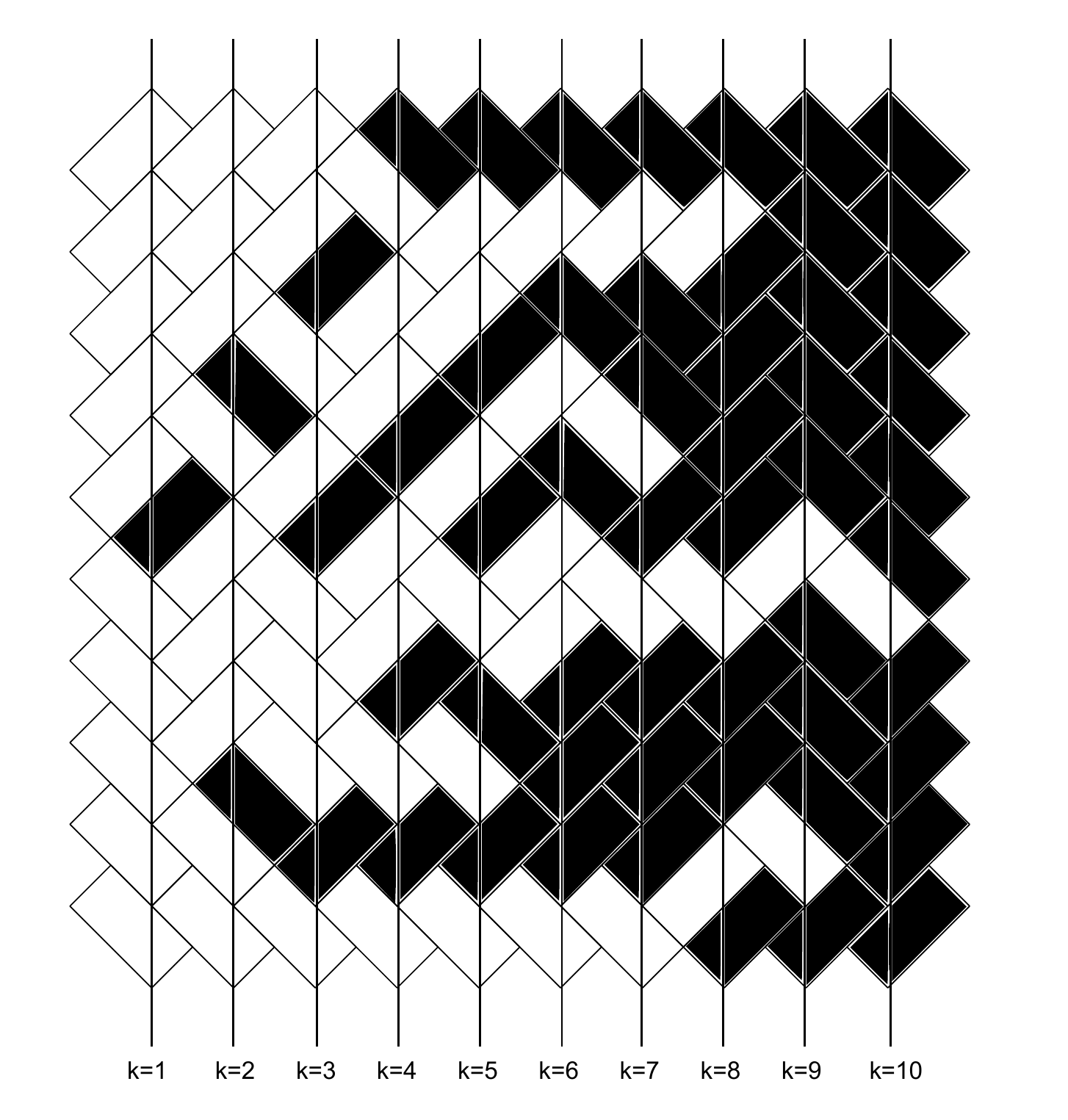}
\caption{An example of the shading of an Aztec diamond of order 10, rotated $45^o$.  Here, the lines pass through the black squares of the checkerboard colouring, and we can see that all the E and S type dominoes (the shaded dominoes) are intersected by lines in their left half, providing an easy way to check the shading. }
\end{figure}

Our interest is in the interlaced particle system implied by a domino tiling of the Aztec diamond.  For its specification, with the Aztec diamond checkerboard coloured as already described, let the horizontal dominoes such that the left square is colour black (white) be called of E (W) type.  Similarly, let the vertical dominoes such that the top square covered is black (white) be called of S (N) type \cite{EKLP92}. Suppose now that the E and S type dominoes are shaded and numbered lines added (see Figure 3).  Each line $k$ passes through the interior of $k$ shaded tiles, and these intersections are considered as specifying the positions of $k$ particles \cite{Jo05a,JN06}. In an appropriate co-ordinate system, these particles occupy distinct positions $x_1^{(k)} > \dots > x_k^{(k)}$ restricted to the lattice points $0,1,2,\dots, N$ on line $k$ ($k=1,\dots, N$).  Most importantly, the particles must satisfy the interlacing condition
\begin{eqnarray}\label{interlace}
x_{i+1}^{(k)} \leq x_i^{(k-1)} \leq x_i^{(k)} & {\rm for} & i=1,\dots, k-1
\end{eqnarray}

A crucial point in relation to our study is the inverse of this mapping. Consider co-ordinate $x_i^{(k-1)}$ on line $k-1$.  Suppose furthermore that the interlacing condition \eqref{interlace} holds with strict inequalities.  Then there are precisely two domino orientations corresponding to $x_i^{(k-1)}$.  On the other hand, if either inequality in \eqref{interlace} is an equality, there is just a single possible domino orientation corresponding to $x_i^{(k-1)}$ (see Figure 4).  Importantly, this means that unlike with the hexagon, given a random tiling of the Aztec diamond with every possibly tiling equally likely, the corresponding particle system must be weighted.

We remark that beyond the theory of tilings of the hexagon and the Aztec diamond, interlaced particle systems with varying numbers of particles occur naturally as the eigenvalues of successive minors of Hermitian random matrix ensembles
\cite{Ba01a,De08a,FN08a}. In fact we will show that in a certain scaling limit the particle system for the Aztec diamond converges in distribution to the eigenvalue process for the minors of Gaussian complex Hermitian random matrices (minors of the GUE ensemble). Also, although not a theme addressed here,
we remark that interlaced particle systems are a rich source of determinantal point processes
\cite{Ra00,FR02,Sa05,JN06,BFPS06, FN08, BP07,BF08,MOW09,No09}.

\begin{figure}[t]\label{Squares}
\includegraphics[scale=0.7]{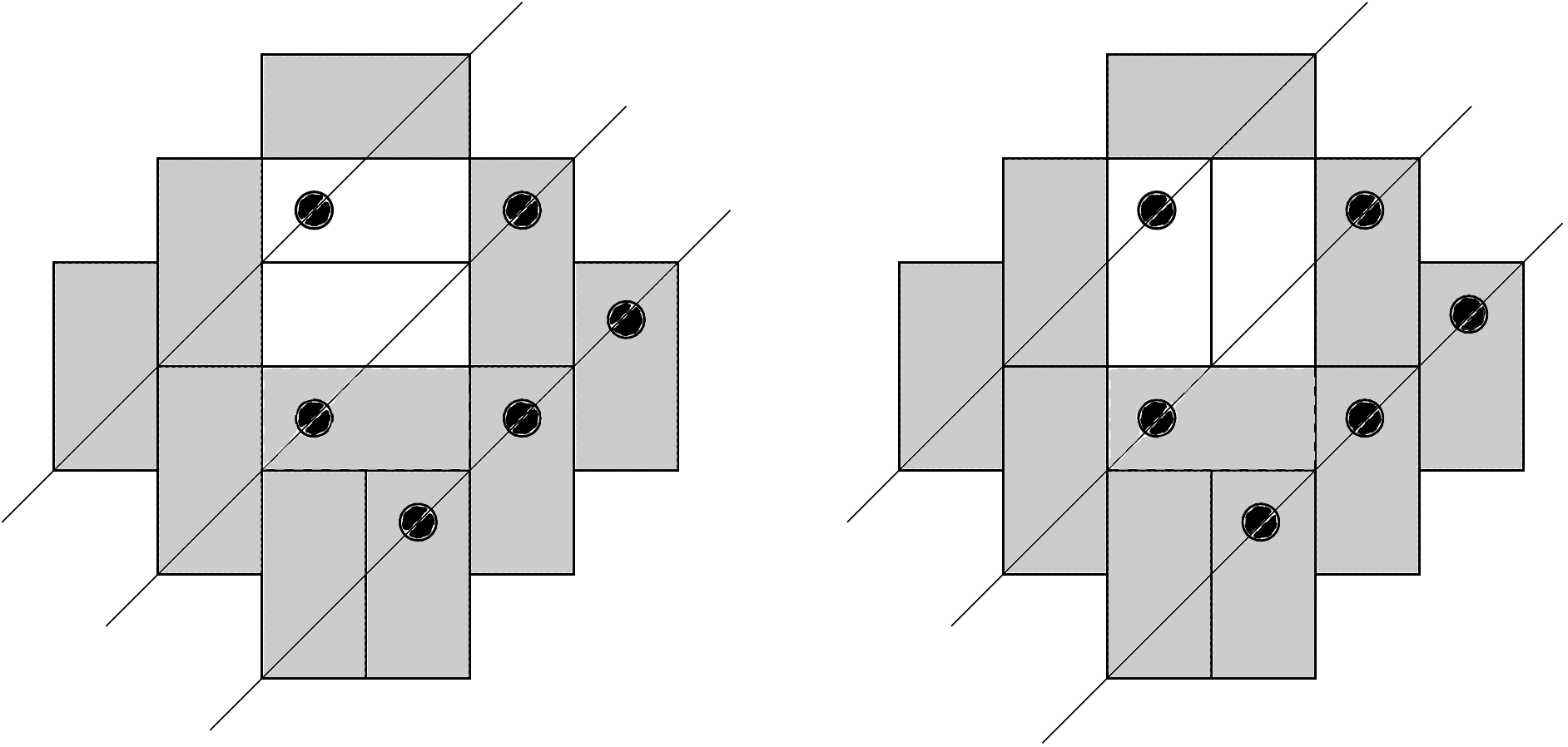}
\caption{An example of two different tilings with the same particle picture. As can be seen, where a particle is not adjacent to another particle on the next line, a `square' is formed, that can be tiled in two different ways.}
\end{figure} 

In this paper, we use the  underlying particle system to rederive fundamental results about random domino tilings of Aztec diamonds, including the number of possible tilings, the multi and single line probability density functions (PDFs) for the positions of shaded particles, the limiting large $N$ shape of the disordered region (Arctic circle effect), and the relation to the GUE minor process from random matrix theory.
We consider a different particle system corresponding to the domino tiling of a half Aztec diamond, and exhibit analogous properties, in particular the limiting large $N$ shape of the disordered region (now half the Arctic circle) and the relation to the anti-symmetric GUE minor process.

\section{One and multi-line PDFs}
Consider a sequence of vertical lines in ${\mathbb R}^2$, with the $k$-th line at $x=k$ and containing $k$ particles.  Let $p(x^{(m)}, \dots, x^{(n)})$ be defined as the joint probability that the $i$-th largest particle on line $k$ is at $(k,x_i^{(k)})$ for $k=m,\dots,n$ and $i=1,\dots, k$.  For the particle system relating to a random tiling of an Aztec diamond of order $N$, we know from the discussion about
(\ref{interlace}) that $x_i^{(k)}$ must obey the restrictions
\begin{align}
&0 \leq x_{i+1}^{(k)} < x_{i}^{(k)} \leq N, \nonumber \\ 
&\label{restrict} x_{i+1}^{(k+1)} \leq x_i^{(k)} \leq x_i^{(k+1)}, \\
&x_i^{(k)} \in {\mathbb Z}, \nonumber
\end{align}
(the second of these is just (\ref{interlace}) with $k \mapsto k+1$).
We also know that although each tiling is equally likely, each particle system is not, since a particle system is not uniquely defined by a tiling.  To account for this, we introduce the notion of adjacency.  Let $x_i^{(k)}$ be called {\em adjacent} if for some $j$, $x_i^{(k)} = x_j^{(k+1)}$.  Each particle that is not adjacent (and is not on the last line) represents a tile with two possible orientations, and so
must be weighted by 2. Given that there are ${1 \over 2} N(N-1)$ particles not on the last line, the 
joint PDF for the entire particle system is given by
\begin{equation}\label{3}
	 p(x^{(1)}, \dots, x^{(N)}) = \frac{2^{N(N-1)/2-\alpha(x^{(1)}, \dots, x^{(N-1))})}}{A_N}\chi(x^{(1)}, \dots, x^{(N)})
\end{equation}
where $\alpha$, the number of adjacent particles, is given by
\begin{equation} \alpha(x^{(m)}, \dots, x^{(n)}) = \sum_{k=m}^n \alpha(x^{(k)}) = \sum_{\substack{i = 1,\dots, k \\ k=m,\dots, n}} \delta_{x^{(k)}_i,x^{(k+1)}_i} + \delta_{x^{(k)}_i,x^{(k+1)}_{i+1}}
\end{equation}
$A_N$ is the number of possible tilings of an Aztec diamond of order $N$, and 
\begin{equation}\label{chi}
\chi(x^{(m)}, \dots, x^{(n)})= \left\{ \begin{array}{ll} 1, & {\rm if} \: \: (x^{(m)}, \dots, x^{(n)}) \;\; {\rm obey } \;\; \eqref{restrict} \\ 0, & {\rm otherwise} \end{array} \right.
\end{equation}
\begin{proposition}\label{pdfprop}
For the particle system corresponding to uniform random tilings of the Aztec diamond of order $N$
\begin{equation}\label{pdfx}
p(x^{(m)}, \dots, x^{(N)}) = \frac{\Delta (x^{(m)})}{D_{m,N} 2^ {\alpha(x^{(m)}, \dots, x^{(N-1)})}} \chi(x^{(m)}, \dots, x^{(N)})
\end{equation}
where 
\begin{equation}
\Delta (x^{(n)}) = \prod_{1\leq i<j \leq n} (x^{(n)}_i-x^{(n)}_j)
\end{equation}
and
\begin{equation}\label{7e}
D_{m,N} = A_N 2^{-N(N-1)/2}\prod_{i=1}^{m-1} i!
\end{equation}
\end{proposition}
\begin{proof}
The $m=1$ case is true from \eqref{3}. Assume the $m=n$ case is true. Then
\begin{equation}
p(x^{(n)}, \dots, x^{(N)}) = \frac{\Delta (x^{(n)})}{D_{n,N} 2^ {\alpha(x^{(n)}, \dots, x^{(N-1)})}}\chi(x^{(n)}, \dots, x^{(N)})
\end{equation}
Summing on the $n$-th line gives 
\begin{eqnarray}
p(x^{(n+1)}, \dots, x^{(N)}) & = &\sum_{x_1^{(n)} = x_2^{(n+1)}}^{x_1^{(n+1)}}\dots\sum_{x_n^{(n)} = x_{n+1}^{(n+1)}}^{x_n^{(n+1)}} \frac{\Delta (x^{(n)})\chi(x^{(n+1)}, \dots, x^{(N)})}{D_{n,N} 2^ {\alpha(x^{(n)}, \dots, x^{(N-1)})}} \nonumber \\
& = & \frac{\chi(x^{(n+1)}, \dots, x^{(N)})}{D_{n,N} 2^ {\alpha(x^{(n+1)}, \dots, x^{(N-1)})}} \det\left[\sum_{t=a_i}^{b_i} \frac{t^{j-1}}{2^{\delta_{t,a_i}+\delta_{t,b_i}}}\right]_{i,j = 1,\dots, n} \nonumber \\
\end{eqnarray}
where we have set $a_i = x_{n-i+2}^{(n+1)}$, $b_i= x_{n-i+1}^{(n+1)}$.
The sum in the determinant is a polynomial function of $a_i$ and $b_i$ with highest degree
term $(b_i^j - a_i^j)/j$. 
Since the lower degree terms will have the same dependence on $a_i, b_i$ for each row $i$, they can be cancelled out by column operations.  Thus
\begin{eqnarray}
p(x^{(n+1)}, \dots, x^{(N)}) & = & \frac{\chi(x^{(n+1)}, \dots, x^{(N)})}{D_{n,N} 2^ {\alpha(x^{(n+1)}, \dots, x^{(N-1)})}} \det\left[\frac{b_i^j-a_i^j}{j}\right]_{i,j={1,\dots, n}} \nonumber \\
& = &  \frac{\Delta (x^{(n+1)})}{n!D_{n,N} 2^ {\alpha(x^{(n+1)}, \dots, x^{(N-1)})}}\chi(x^{(n+1)}, \dots, x^{(N)})
\end{eqnarray}
where the determinant evaluation follows by noting that it must contain $\Delta (x^{(n+1)})$ as a
factor, and is of the same degree as $\Delta (x^{(n+1)})$. The case $m=n+1$ has thus been established, provided 
$D_{n+1,N} = n!D_{n,N}$, which is indeed a property of (\ref{7e}).
\end{proof}
To find $A_N$, we introduce virtual particles $\{x_i^{(N+1)}\}_{i=1,\dots N+1}$ to the system with the
requirement that also obey \eqref{restrict}.  We note that the only possibility is $x^{(N+1)} = \{0,1,\dots, N\}$.  Following the method of the proof of Proposition \ref{pdfprop},  but beginning with the PDF
\begin{equation}
p(x^{(1)}, \dots, x^{(N+1)}) = \frac{2^{N(N+1)/2-\alpha(x^{(1)}, \dots, x^{(N)})}}{A_N}\chi(x^{(1)}, \dots, x^{(N+1)})
\end{equation}
(the term $N(N+1)/2$ in the exponent results from there now being $N+1$ lines), we end up with the single line PDF
\begin{equation}\label{onelineN1}
p(x^{(N+1)}) = \frac{2^{N(N+1)/2}\Delta (x^{(N+1)})}{A_N \prod_{i=1}^N i!}
\end{equation}
But $x_i^{(N+1)}=N+1-i$, which substituted in \eqref{onelineN1} gives
\begin{equation}
p\left(x^{(N+1)} = \{0,1,\dots, N\}\right)  =  \frac{2^{N(N+1)/2} \prod_{0\leq i < j \leq N} (j-i)}{A_N \prod_{m=1}^N m!} = 1
\end{equation}
Hence we conclude
\begin{equation}
A_N  =  2^{N(N+1)/2} \label{CN}
\end{equation}
The result \eqref{CN} for the number of domino tilings of the Aztec diamond was first derived by \cite{EKLP92}.  Since then a number of derivations distinct from those given in \cite{EKLP92} have been found, for example \cite{Jo02,EF05}.  The present derivation using the particle picture appears to be new.

It remains to compute the single line PDF, which is gotten from (\ref{pdfx}) by summing over
$x^{(m+1)},\dots, x^{(N)}$. To perform the summations we
introduce a second set of particles $y_i^{(k)}$ representing the unshaded tiles.  As with the shaded tiles, we want the $k$-th line to have $k$ particles, so we label the lines from right to left in the $y$ picture.  Because every position must have a shaded or unshaded tile, and no position can have both, the $y^{(k)}$ are defined such that 
\begin{eqnarray}
x^{(n)} \cup y^{(N+1-n)} & = & \{0,1,\dots N\} \nonumber\\
x^{(n)} \cap y^{(N+1-n)} & = & \emptyset. \label{xy}
\end{eqnarray}
Because the unshaded tiles, when viewed right to left, obey the same probabilistic law as the black tiles view left to right, the formulas for
$p(x^{(m)},\dots, x^{(n)})$ and  $p(y^{(m)},\dots, y^{(n)})$ are the same for all $m,n$.  

We now express \eqref{pdfx} as a function of the $y_i^{(k)}$.  To begin, using the fact that\begin{equation}
\Delta(\{0,1,\dots N\}) = \prod_{i=1}^N i!
\end{equation}
we have \begin{equation} \label{deltaswap}
\Delta(x^{(n)}) = \frac{\Delta(y^{(N+1-n)})\prod_{i=1}^N i!}{\prod_{i=1}^{N+1-n}y_i^{(N+1-n)}! (N-y_i^{(N+1-n)})!}.
\end{equation}
It remains to calculate $2^{\alpha(x^{(m)}, \dots, x^{(N-1)})}$ in terms of $y_i^{(k)}$.  
\begin{proposition}\label{alpha}
Consider two lines $n$, $n+1$ in a particle system as defined above, but generalised so that there are $N^*$ possible positions for particles on each line. Let this system be filled with $x$ and $y$ particles, such that every possible position has either an $x$ or $y$ particle, and no position has both.  Furthermore let the lines labels be changed to $n^*$ and $n^*-1$ respectively when considering the $y$ particles.  If line $n$ has $a$ $x$-particles $\{x_1^{(n)} ,\dots, x_a^{(n)}\}$ (and therefore $N^*-a$ $y$-particles $\{y_1^{(n^*)} ,\dots, y_{N^*-a}^{(n^*)}\}$) and line $n+1$ has $b$ $x$-particles $\{x_1^{(n+1)} ,\dots, x_b^{(n+1)}\}$ (and therefore $N^*-b$ $y$-particles $\{y_1^{(n^*-1)} ,\dots, y_{N^*-b}^{(n^*-1)}\}$)then,  for $\alpha$ defined as above, 
\begin{equation}\label{18}
\alpha(x^{(n)}) =  \alpha(y^{(n^*-1)}) +a+b-N^*
\end{equation}
\end{proposition}
\begin{proof}
There are $a$ $x$'s on line $n$.  Of these, $\alpha(x^{(n)})$ are adjacent.  Therefore, exactly $a-\alpha(x^{(n)})$ of the $x$'s on line $n$ are not adjacent. Noting that line $n+1$ is the lower numbered line in the $y$ picture, this means that exactly $a-\alpha(x^{(n)})$ of the $y$'s on line $n+1$ are not adjacent. Since there are, by definition of $\alpha$, $N^*-b-\alpha(y^{(n^*-1)})$ non-adjacent $y$'s on line $n+1$, 
\begin{equation}
a-\alpha(x^{(n)}) =  N^*-b-\alpha(y^{(n^*-1)}) 
\end{equation}
and so (\ref{18}) follows.
\end{proof}
Using Proposition \ref{alpha} with $a=n$, $b=n+1$, $N^*=N+1$ and $n^*=N+1-n$ gives
\begin{eqnarray}
\alpha(x^{(n)}) &= & \alpha(y^{(N-n)}) +2n-N \\
\label{alphaswap1} \alpha(x^{(m)},\dots,x^{(N-1)}) & = & \alpha(y^{(1)}, \dots, y^{(N-m)}) + (N-m)(m-1)
\end{eqnarray}
So, applying \eqref{deltaswap} and \eqref{alphaswap1} to  \eqref{pdfx} we have
\begin{equation}\label{1-mpdf}
p(y^{(1)}, \dots, y^{(n)}) = \frac{\prod_{i=0}^{n-1} (N-i)!}{\prod_{i=1}^{n}y_i^{(n)}! (N-y_i^{(n)})!}\frac{\Delta(y^{(n)})\chi(y^{(1)}, \dots, y^{(n)})}{2^{N+(N-n)(n-1) + \alpha(y^{(1)}, \dots, y^{(n-1)})}}
\end{equation}
But we know $p(x^{(1)}, \dots, x^{(n)}) = p(y^{(1)}, \dots, y^{(n)})$. Finally, using the same inductive method from the proof of Proposition  \ref{pdfprop}, we end up with
\begin{equation}\label{1line}
p(x^{(n)}) = \frac{\Delta(x^{(n)})^2}{2^{N+(N-n)(n-1)} \prod_{i=1}^{n}x_i^{(n)}! (N-x_i^{(n)})!}\prod_{i=0}^{n-1}\frac{(N-i)!}{i!}
\end{equation}
This has been derived using different arguments in \cite{Jo01x} (in particular the weighted
particle system is not specified by (\ref{3})), where it is 
recognised as a particular example of a discrete orthogonal polynomial unitary
ensemble based on a the  Krawtchouk weight with $p=1/2$) (see the Appendix).

\section{Large $N$ limits}
In \cite{JN06} the weighted particle process corresponding to an Aztec diamond tiling,
defined through its correlations and restricted to the first $n$ lines, was shown in a certain
scaling limit to coincide with the minor process for a certain ensemble of random matrices.
These random matrices are the ensemble of complex Gaussian matrices $X$with measure
proportional to $e^{-{\rm Tr} \, X^2/2}$, to be denoted GUE${}^*$ (conventionally
the GUE ensemble has measure proportional to $e^{-{\rm Tr} \, X^2}$). The minor process is formed
out of the correlated eigenvalues $\cup_{j=1}^n \{ z^{(j)} \}$, $z^{(j)} = (z_1,z_2,\dots,z_n)$ denoting the
eigenvalues of the $j$-th minor. This is known \cite{Ba01a} to have joint PDF 
\begin{equation}\label{gminor}
{1 \over C_n} \prod_{l=1}^{n} e^{- (z_l^{(n)})^2/2} \prod_{1 \le j < k \le n}
(z_j^{(n)} - z_k^{(n)}) \prod_{j=1}^{n-1} \chi(z^{(j+1)} > z^{(j)})
\end{equation}
where, with $\chi_A$ the indicator function for the set $A$,
$$
\chi(z^{(j+1)} > z^{(j)}) := \chi_{z_1^{(j+1)} > z_1^{(j)} > \cdots > z_j^{(j+1)} >
z_j^{(j)} > z_{j+1}^{(j+1)}}
$$
and the normalization $C_n$ is given by $C_n = (2 \pi)^{n/2}$.
We can show directly that that the joint PDF for the weighted particle process tends to
(\ref{gminor}) in an appropriate limit.

\begin{proposition}\label{large N}
Let the points $z_i^{(j)} := (2y_i^{(j)}-N)/\sqrt{N}$ be a rescaling of the points $y_i^{(j)}$, where $N$ is the order of the Aztec diamond as described above.  Given that the $y_i^{(j)}$ have PDF $p$ as described in \eqref{1-mpdf}, one has\begin{eqnarray}
p(y^{(1)}, \dots, y^{(n)}) \rightarrow p^*(z^{(1)}, \dots, z^{(n)}) & as & N \rightarrow \infty \end{eqnarray} where $p^*$ is the PDF  for the GUE${}^*$ minor process as specified by \eqref{gminor}.
\end{proposition}
\begin{proof}
Let $y = (z\sqrt{N}+ N)/2 :=g(z)$.  Then 
\begin{equation}\label{pN}
p(y^{(1)}, \dots, y^{(n)})  = p\left(g(z^{(1)}), \dots, g(z^{(n)})\right)
\prod_{i=1}^j \prod_{j=1}^n  g'(z_i^{(j)})
\end{equation}
Clearly, $g'(z_i^{(j)}) = \sqrt{N}/2$, so 
\begin{equation}
\prod_{\substack{i=1,\dots, j \\ j=1,\dots, n}} g'(z_i^{(j)}) = \Big ( \frac{N}{4} \Big )^{n(n+1)/4}
\end{equation}
We now wish to compute $p\left(g(z^{(1)}), \dots, g(z^{(n)})\right)$ in the limit $N \rightarrow \infty$.  Applying forms of Stirling's approximation (for large $N$)
\begin{eqnarray*}
\label{SA1}(aN+b)! & \sim & \sqrt{2\pi a N} (a N)^{a N + b} e^{-a N} \\
\label{SA2}(aN +b\sqrt{N} +c)! & \sim & \sqrt{2\pi a N} (a N)^{aN +b\sqrt{N} +c} e^{b^2/2a-aN} 
\end{eqnarray*}
to \eqref{1-mpdf}, and using \eqref{pN} we have
\begin{equation}
p(y^{(1)}, \dots,y^{(n)}) \sim  \frac{\Delta(z^{(n)}) }{(2\pi)^{n/2}} e^{ -\sum_{i=1}^n \frac{1}{2}(z_i^{(n)})^2}  \prod_{j=1}^{n-1} \chi(z^{(j+1)} > z^{(j)})
\end{equation}
which is \eqref{gminor}.
\end{proof}

We would also like to compute the region of support for this particle system.  In the Aztec diamond tiling this corresponds to the boundary of the disordered region.
The region of support on any given line $n$ is the interval $[a_n,b_n]$ in which, for a large enough number of particles, all the particles will lie within that region with probability $1$.  Since we are taking $n$, the number of particles, to be large and $n \leq N$ we must take $N$, the number of lines, to be large also. It thus makes sense to scale $n$ so that the `line label' $s= n/N$ is a real number in $[0,1]$. The region of support of the system will be the areas in between the graphs of $a(s)$, $b(s)$.  To calculate the functional form of these boundaries from \eqref{1line}, one approach would be to use the fact that this PDF relates to the  Krawtchouk ensemble. The necessary details have been
given in \cite{JPS98}. Here we give a more physically motivated derivation, based on a log-gas picture \cite{Fo10}.

The Boltzmann factor for a log-gas of $N_p$ particles has the form 
\begin{equation}\label{loggas}
\prod_{1\leq i < j \leq N_p} |x_i-x_j|^\beta \prod_{k=1}^{N_p} e^{-\beta V(x_k)}
\end{equation}
where $\beta$ denotes the inverse temperature and $V(x)$ is a one body potential, due to background charge density $-\rho(x)$.  Explicitly,
\begin{equation}\label{V}
V(x) := \int_a^b \rho(t) \log |t-x| dt.
\end{equation}
A hypothesis of the log-gas picture is that for large $N_p$ and to leading order the particle charge density and background charge density cancel, so that the particle density is to leading order
equal to $\rho(x)$. 

In the cases that $\rho(x)$ is supported on a single interval $[a,b]$ (which we expect for the log-gas interpretation of \eqref{1line}), normalization of the density requires
\begin{equation}\label{intdens}
\int_a^b \rho(t) dt = N_p.
\end{equation}
Furthermore, the explicit form of $\rho(x)$ obtained by solving the integral equation \eqref{V} is known
in terms of $V(x)$, and the boundary of the support is determined by the equations \cite{Fo10}
\begin{eqnarray}
\label{support0}  \int_a^b \frac{V'(t)}{\sqrt{(b-t)(t-a)}} dt & = & 0 \\
\label{support1} \int_a^b \frac{tV'(t)}{\sqrt{(b-t)(t-a)}} dt & = & \pi N_p
\end{eqnarray}

As written, \eqref{1line} is a lattice gas variant of the log-gas \eqref{loggas} in the case $\beta = 2$.  In the limit $n \rightarrow \infty$, the lattice gas approaches the continuum log-gas upon the substitution
\begin{equation}
x_i^{(n)} = Nt_i^{(n)}
\end{equation}
where, to leading order in $N$, $0 \leq t_i^{(n)} \leq 1$.  In terms of the co-ordinate $t_i^{(n)} = t_i$,  the one body factor in \eqref{loggas} reads
\begin{equation}
e^{-2V(t)} = \frac{1}{(Nt)!(N-Nt)!}
\end{equation}
and recalling $s=N/n$ shows $N_p = Ns$.
Thus solving \eqref{support0} and \eqref{support1} in the limit $N\rightarrow \infty$  gives $(a(s), b(s))$, the support in the variable $t$.

From \eqref{restrict} we know that $a, b \in [0,1]$.  Noting that $V(t) = V(1-t)$ and inserting this into \eqref{support0}, we have that for some $c \in [0,\frac{1}{2}]$, $a=\frac{1}{2}-c$, $b=\frac{1}{2} +c$.  Computing
\begin{equation}
\lim_{N\rightarrow \infty} \frac{2V'(y)}{N} = \log\left(\frac{y}{1-y}\right)
\end{equation}
leaves us with 
\begin{equation}
\int_{\frac{1}{2}-c}^{\frac{1}{2}+c} \frac{t \log(t/(1-t))}{\sqrt{(\frac{1}{2}+c-t)(t+c-\frac{1}{2})}} dt = 2\pi s
\end{equation}
to solve for $c=c(s)$.  The change of variables $t = \frac{1}{2}+u$ leads us to a more managable
\begin{equation}\label{intcs}
\int_{-c}^c \frac{u \log(1+2u)}{\sqrt{c^2-u^2}}du = \pi s.
\end{equation}
The integral can be computed exactly (a computer algebra package was used), giving
\begin{equation}
1- \sqrt{1-4c^2} = 2s.
\end{equation}
Recalling that $c \in [0, \frac{1}{2}]$, we see that this has a solution for $s \in [0, \frac{1}{2}]$ only.
For $s \in (\frac{1}{2}, 1]$, physical interpretation of the relationship between $s$ and $c$ (increasing $s$, the number of particles, must not decrease $2c$, the size of their support) leads us to define $c(s) :=\frac{1}{2}$ for $s \in (\frac{1}{2}, 1]$.  So we have
\begin{equation}\label{fr}
a(s) = \left\{\begin{array}{ll} \frac{1}{2}(1-\sqrt{1-(1-2s)^2}) & s \in [0,\frac{1}{2}] \\ 0 & s \in (\frac{1}{2},1]\end{array} \right.
\end{equation}
\begin{equation}
b(s) = \left\{\begin{array}{ll} \frac{1}{2}(1+\sqrt{1-(1-2s)^2}) & s \in [0,\frac{1}{2}] \\ 1 & s \in (\frac{1}{2},1]\end{array} \right.
\end{equation}

Noteworthy here is that, because of the mirrored nature of the shaded and unshaded tiles, the area of support of the unshaded tiles is related to the area of support of the shaded tiles by
\begin{eqnarray}
a_{\rm shaded}(s) = a_{\rm unshaded}(1-s), & & b_{\rm shaded}(s) = b_{\rm unshaded}(1-s),
\end{eqnarray}
so the disordered region of the Aztec diamond tiling, the area that has both shaded and unshaded tiles, is a 
perfect circle --- the Arctic circle --- in the limit $N\rightarrow \infty$ \cite{JPS98}.

\section{The half Aztec diamond}
Consider an Aztec diamond of order $N=2(M+1)$ rotated by forty five degrees as in
Figure \ref{Aztec Shading}.  Define a restriction on the tiling of this Aztec diamond such that in the particle picture as defined above, a particle at $x$ on line $j$ implies no particle at $x$ on line $N+1-j$.  Because of the interlacing restriction, this means that in the tiling picture
the whole middle column between lines $k=M+1$ and $M+2$ will consist of
squares formed from a pair of dominoes rotated $45^o$.  If we delete all these squares we are left with two halves.  We will call these half Aztec diamonds of order $M$.  The present half
Aztec diamond model bears some resemblance to the Aztec diamond with barriers
introduced in \cite{PS99a}.

By construction, the tiling corresponding to two half Aztec diamonds are mirror images. We will call any
tiling of an Aztec diamond of order $N = 2(M+1)$ formed from two half Aztec diamonds symmetric.
With $C^*_N$ the number of symmetric tilings of an Aztec diamond of order $N$ and
$H_M$ the number of tilings of a half Aztec diamond of order $M$, we therefore have
\begin{equation} 
C^*_{2(M+1)} = H_M 2^{M+1}
 \end{equation} 
 Here the factor of $2^{M+1}$ corresponds to the number of tilings of the deleted squares. 
 
 We would like to use the particle picture to compute $H_M$. We begin by noting that the joint
 PDF for the weighted particle system is
\begin{equation}\label{46}
p(x^{(1)},\dots, x^{(M+1)}) = \frac{2^{M(M+1)/2 -\alpha(x^{(1)}, \dots, x^{(M)})}}{H_M}\chi(x^{(1)}, \dots, x^{(M)})
\end{equation}
(cf.~(\ref{3})), 
with the additional restriction that 
\begin{equation}\label{46a}
x_i^{(M+1)} = 2M +3-2i
\end{equation}
 Using the method of derivation of (\ref{pdfx}) it follows from this that
\begin{equation}
p(x^{(M+1)}) =  \frac{2^{M(M+1)/2}\Delta (x^{(M+1)})}{H_M \prod_{m=1}^M m!}
\end{equation}
which must equal $1$ for 
\begin{equation}\label{46a}
x^{(M+1)} = \{1, 3, \dots, 2M+1\} 
\end{equation}
Consequently 
\begin{equation}\label{H}
H_M = 2^{M(M+1)}
\end{equation}
 and $C^*_{2N} = 2^{N^2}$. Note that
$$
\lim_{N \to \infty} {1 \over N^2} \log A_N = 2 \lim_{N \to \infty} {1 \over N^2}
\log H_M \Big |_{M = N/2 - 1}
$$
as to be expected from the interpretations of these quantities as entropies for the tiling problem.

There is a second particle system associated with symmetric tilings. This is obtained by rotating the
half Aztec diamond --- which has $M$ vertical lines --- by $90^o$ to obtain a half Aztec diamond
positioned with long side horizontal and thus having $N = 2(M+1)$ vertical lines (recall
Figure \ref{Aztec Shading}). The first of these is empty of particles and last one is full. Ignoring
these two lines we have $2M$ lines where successive lines $2n-1$ and $2n$ ($n=1,\dots,M$)
have $n$ particles. We would like to develop the properties of this particle system.

Analogous to (\ref{46}), although with $H_M$ substituted by its evaluation (\ref{H}), the joint PDF
for this weighted particle system is
 \begin{equation}\label{half pdf}
p(x^{(1)},\dots, x^{(2M)}) = \frac{\chi^{*}(x^{(1)}, \dots, x^{(2M)})}{2^{M+\alpha(x^{(1)},\dots, x^{(2M-1)})}}
\end{equation}
where $\chi^*$ is the same as $\chi$ in \eqref{chi}, except the first restriction is changed to {\small $1\leq x_{i+1}^{(j)} < x_i^{(j)} \leq M+1$}.
\begin{proposition}\label{pp}
Let
\begin{equation}\label{HH}
H_{2m-1,M} = 2^{M}\prod_{i=1}^{m-1} (2i)! \qquad H_{2m,M} = \frac{2^M\prod_{i=1}^{m} (2i)!}{2^mm!}
\end{equation}
For $p$ as defined in \eqref{half pdf},
\begin{eqnarray}
p(x^{(2n-1)},\dots, x^{(2M)}) & = & \frac{\Delta(x^{(2n-1)}) S(x^{(2n-1)})}{H_{2n-1,M} 2^{\alpha(x^{(2n-1)},\dots, x^{(2M-1)})}} \label{odd}\\ 
p(x^{(2n)},\dots, x^{(2M)}) & = & \frac{\Delta(x^{(2n)}) S(x^{(2n)})\prod_{i=1}^n( x_i^{(2n)} -\frac{1}{2})}{H_{2n,M} 2^{\alpha(x^{(2n)},\dots, x^{(2M-1)})}} \label{even}
\end{eqnarray}
where $S(x^{(n)}) := \prod_{1\leq i<j\leq n} (x_i^{(n)} +x_j^{(n)}-1)$
\end{proposition}
\begin{proof}
We proceed as in the proof of Proposition \ref{pdfprop}.
The $2n-1=1$ case is true from \eqref{half pdf}.
Assume the $n=2m-1$ case is true. Then summing on the $(2m-1)$-th line gives
\begin{eqnarray}
p(x^{(2m)}, \dots, x^{(2M)}) & = &\sum_{x_1^{(2m-1)} = x_2^{(2m)}}^{x_1^{(2m)}}\dots\sum_{x_m^{(2m-1)} = 1}^{x_m^{(2m)}} \frac{\Delta(x^{(2m-1)}) S(x^{(2m-1)})}{H_{2m-1,M} 2^{\alpha(x^{(2m-1)},\dots, x^{(2M-1)})}} \nonumber \\
& = & \frac{1}{H_{2m-1,M} 2^ {\alpha(x^{(2m)}, \dots, x^{(2M-1)})}} \det\left[d_{i,j}\right]_{i,j=1,\dots,m}
\end{eqnarray}
where, with $a_i = x_{m-i+2}^{(2m)}$, $b_i= x_{m-i+1}^{(2m)}$
\begin{align}
d_{i,j} = \left\{\begin{array}{ll} \sum_{t=1}^{b_1}   2^{-\delta_{t,b_1}}  (t-\frac{1}{2})^{2(j-1)}, & i=1 \\ 
\sum_{t=a_i}^{b_i} 2^{-\delta_{t,a_i}-\delta_{t,b_i}} (t-\frac{1}{2})^{2(j-1)} {}, & i=2,\dots,m \end{array}\right.
\end{align}
This implies
\begin{equation}
p(x^{(2m)}, \dots, x^{(2M)}) = \frac{2^mm!\prod_{i=1}^m (x_i^{(2m)} -\frac{1}{2})}{H_{2m-1,M}(2m)! 2^{\alpha(x^{(2m)}, \dots, x^{(2M-1)})}}\det\left[ \Big (x_{m-i-1}^{(2m)}-\frac{1}{2} \Big )^{2(j-1)}\right]_{i,j=1,\dots,m}
\end{equation}
which recalling (\ref{HH}) establishes the case $n=2m$.  Summing now on line $2m$ gives
\begin{equation}
p(x^{(2m+1)}, \dots, x^{(2M)}) = \frac{\Delta(x^{(2m+1)}) S(x^{(2m+1)})
}{2^m m!H_{2m,M} 2^{\alpha(x^{(2m+1)},\dots, x^{(2M-1)})}} 
\end{equation}
and this  recalling (\ref{HH}) establishes the case $n=2m+1$. 
\end{proof}
We want to use Proposition \ref{pp} to deduce the one line PDFs.
For this we again introduce particles $y_i^{(j)}$ representing all lattice sites not occupied by an $x$ particle:
\begin{eqnarray}
x^{(n)} \cup y^{(2M+1-n)} & = & \{1,\dots M+1\} \\
x^{(n)} \cap y^{(2M+1-n)} & = & \emptyset.
\end{eqnarray}
Using the fact that
\begin{equation}
S(\{1,\dots, M+1\}) = \prod_{i=1}^{M}\frac{(2i)!}{i!}
\end{equation}
we have
\begin{equation}
S(x^{(n)}) = S(y^{(2M-n+1)}) \prod_{i=1}^r(2y^{(2M-n+1)}_i-1) \prod_{i=1}^r \frac{(y^{(2M-n+1)}_i-1)!}{(y^{(2M-n+1)}_i+M)!}\prod_{j=1}^{M}\frac{(2j)!}{j!}
\end{equation}
where $r = |y^{(2M-n+1)}|$.  It is also not hard to compute that
\begin{equation}
\prod_{i=1}^n(x_i^{(2n)} -\frac{1}{2}) = \frac{(2M+2)!}{2^{2M+2}(M+1)!} \frac{1}{\prod_{i=1}^{M-n+1} (y_i^{(2(M-n+1)-1)}-\frac{1}{2})}
\end{equation}
Using Proposition \ref{alpha} we have that
\begin{equation}
\alpha (x^{(n)})  =  n -M+ \alpha(y^{(2M-n)}) 
\end{equation}
so, changing from $x$ to $y$ in \eqref{odd} and \eqref{even} respectively we obtain
\begin{eqnarray}\label{66}
p(y^{(1)},\dots, y^{(2m-1)}) & = & \frac{(M+1-m)!}{(M+1)!}\frac{\prod_{j=0}^{m-1}(2(M-m+j+2))!}{2^{\alpha(y^{(1)},\dots, y^{(2m-2)})+2m(M-m+1)+m}} \nonumber \\
\label{oddy}& &\times \;\;\;\;\; \frac{\Delta(y^{(2m-1)})S(y^{(2m-1)})}{\prod_{i=1}^{m}(y^{(2m-1)}_i+M)!(M+1-y_i^{(2m-1)})!}\\  
p(y^{(1)},\dots, y^{(2m)}) & = &\frac{\prod_{j=1}^m (2(M-m+j))!}{2^{\alpha(y^{(1)},\dots, y^{(2m-1)})+2m(M-m)}} \nonumber\\
\label{eveny} & & \times \;\;\;\;\; \frac{\Delta(y^{(2m)})S(y^{(2m)})\prod_{i=1}^m(y^{(2m)}_i-\frac{1}{2})}{\prod_{i=1}^{m}(y^{(2m)}_i+M)! (M+1-y_i^{(2m)})!}
\end{eqnarray}
Using the method of the proof of Proposition \ref{pp} we compute from these that
the one-line PDFs are
\begin{eqnarray}\label{67}
p(y^{(2m-1)}) & = & \frac{(M+1-m)!}{(M+1)!}\prod_{j=0}^{m-1}\frac{(2(M-m+j+2))!}{(2j)!}\frac{1}{2^{2m(M-m+1)+ m}} \nonumber\\
\label{odd1line}& & \times \;\;\;\;\; \frac{\Delta(y^{(2m-1)})^2S(y^{(2m-1)})^2}{\prod_{i=1}^{m}(y^{(2m-1)}_i+M)!(M+1-y_i^{(2m-1)})!} \\
p(y^{(2m)}) & = & \frac{m!}{2^{2m(M-m)-m}}\prod_{j=1}^m \frac{(2(M-m+j))!}{(2j)!}\nonumber \\
\label{even1line}& & \times \;\;\;\;\; \frac{\Delta(y^{(2m)})^2S(y^{(2m)})^2\prod_{i=1}^m(y^{(2m)}_i-\frac{1}{2})^2}{\prod_{i=1}^{m}(y^{(2m)}_i+M)! (M+1-y_i^{(2m)})!}\label{even1line}
\end{eqnarray}

We remarked above that the one-line PDF (\ref{1line}) corresponds to a discrete orthogonal polynomial unitary ensemble based on a particular Krawtchouk weight. As detailed in the Appendix, (\ref{66}) and
(\ref{67}) may be regarded as type $B$ versions of the same ensembles.

We saw in Proposition \ref{large N} that the particle system for the Aztec diamond in the large
$N$ limit relates to the
GUE minor process. This is also true of the particle system  associated with rhombi tiling
of the hexagon revised in the Introduction \cite{JN06}. In the case of rhombi tiling of a half
hexagon $(2a,b,b)$, cut horizontally along the side $2a$ (recall Figure \ref{Hexagon}), it
is shown in \cite{FN08a} that in the large $a$ limit the particle process converges to the
eigenvalue process for the minors of anti-symmetric GUE matrices.
Here we will show that this remains true of the particle system for the half Aztec diamond.

The anti-symmetric GUE is the probability on purely imaginary $n \times n$
Hermitian matrices with measure
proportional to $e^{-{\rm Tr} \, X^2/2}$. Using the same notation as in (\ref{gminor}),
we know from \cite{FN08a} that the joint PDF for the positive eigenvalues of the minor is given
by
\begin{equation}\label{agminor1}
{1 \over C_n} \prod_{l=1}^{n/2} e^{- (x_l^{(n)})^2} \prod_{1 \le j < k \le n/2}
\left((x_j^{(n)})^2 - (x_k^{(n)})^2\right) \prod_{j=1}^{n-1} \chi(x^{(j+1)} > x^{(j)}),
\end{equation}
for $n$ even and
\begin{equation}\label{agminor2}
{1 \over C_n} \prod_{l=1}^{(n-1)/2} x_l^{(n)}e^{- (x_l^{(n)})^2} \prod_{1 \le j < k \le (n-1)/2}
\left((x_j^{(n)})^2 - (x_k^{(n)})^2\right) \prod_{j=1}^{n-1} \chi(x^{(j+1)} > x^{(j)}),
\end{equation}
for $n$ odd. Here
$$
C_n = \left \{\begin{array}{ll} \pi^{n/4} 2^{-n^2/4}, & n \: {\rm even} \\
  \pi^{(n-1)/4} 2^{-n(n-1)/4}, & n \: {\rm odd}  \end{array}  \right.
$$
and it is understood that $x^{(1)} = 0$. A straight forward limiting procedure applied to \eqref{oddy} and \eqref{eveny}, 
according to the strategy of the proof of Proposition \ref{large N} gives convergence
to these PDFs.

\begin{proposition}\label{large N half}
Let the points $z_i^{(j+1)} := y_i^{(j)}/ \sqrt{M}$ be a rescaling of the points $y_i^{(j)}$, where $M$ is the order of the half Aztec diamond as described above, and let $z^{(1)}_1=0$.  Given that the $y_i^{(j)}$ have PDF $p$ as described in \eqref{oddy} and \eqref{eveny}, 
then \begin{eqnarray}
p(y^{(1)}, \dots, y^{(n)}) \rightarrow p_{aGUE}(z^{(1)}, \dots, z^{(n)}) & as & N \rightarrow \infty \end{eqnarray} where $p_{aGUE}$ is the PDF  for the anti-symmetric GUE minor process specified by (\ref{agminor1})
and (\ref{agminor2}).
\end{proposition}

Our last task is to compute the limiting support.
We proceed using the same method as for the full Aztec diamond case.  Here however, we will be dealing with Boltzmann factors of the form
\begin{eqnarray}\label{loggas2}
\prod_{1\leq i < j \leq N_p} |x_i^2-x_j^2|^\beta \prod_{k=1}^{N_p} e^{-\beta V(x_k)}
\end{eqnarray}
so \eqref{V} becomes,
\begin{equation}\label{V2}
V(x) := \int_0^a \rho(t) \log|x^2-t^2| dy.
\end{equation}
If we define $\rho(-x):= \rho(x)$, then this can be expressed
\begin{equation}
V(x) = \int_{-a}^a \rho(y) \log|x-y| dy
\end{equation}
similar to \eqref{V}, although now \eqref{intdens} becomes 
\begin{equation}\label{distint2}
\int_{-a}^a \rho(y) dy = 2N_p.
\end{equation}
In the limit $M\rightarrow \infty$, \eqref{odd1line} and \eqref{even1line} approach a continuum log-gas \eqref{loggas2} upon the substitution
\begin{equation}
x_i^{(n)} - \frac{1}{2} = Mt_i^{(n)}
\end{equation}
where, to leading order, $0 \leq t \leq 1$.

In terms of the co-ordinate $t_i = t_i^{(2n-1)}$, the one body factor in \eqref{loggas2} reads
\begin{equation}\label{oddV}
e^{-2V(z)} = \frac{1}{\left(M(1+t)+\frac{1}{2}\right)!\left(M(1-t)+\frac{1}{2}\right)!}
\end{equation}
Let $s = n/M$. Then $N_p = Ms$ and from \eqref{oddV}
\begin{equation}\label{Vdash}
\lim_{M\rightarrow \infty} \frac{2V'(t)}{M} = \log\left(\frac{1+t}{1-t}\right).
\end{equation}
According to \eqref{support1} but taking into account \eqref{distint2}, $a(s)$ is given by solving
\begin{equation}
\int_{-a}^a \frac{t \log\left(\frac{1+t}{1-t}\right)}{\sqrt{a^2-t^2}} \, dt= 4\pi s,
\end{equation}
(cf.~\eqref{intcs}) and the integral can be evaluated to give
\begin{equation}
1- \sqrt{1-a^2} = 2s
\end{equation}
This can be solved immediately for $s \in [0,\frac{1}{2}]$.
However, for $s \in (\frac{1}{2}, 1]$, like with the earlier case, this equation has no solution for
 $a \in [0,1]$.  Using the same logic as leading to (\ref{fr}), we define $a(s) := 1$ for $s \in (\frac{1}{2},1]$, so
\begin{equation}\label{80}
a(s) = \left\{\begin{array}{ll} \sqrt{1-(1-2s)^2} & s \in [0,\frac{1}{2}] \\ 1 & s \in (\frac{1}{2},1]\end{array} \right.,
\end{equation}
giving the same shape as the top half of the full Aztec diamond, as expected.

In terms of the co-ordinate $t_i = t_i^{(2n)}$, the one body factor in \eqref{loggas2} reads
\begin{equation}
e^{-2V(t)} = \frac{Mt}{\left(M(1+t)+\frac{1}{2}\right)!\left(M(1-t)+\frac{1}{2}\right)!}
\end{equation}
This gives the same equation for $V'(z)$ as in \eqref{Vdash}, and thus the same result (\ref{80})
for $a(s)$, again as expected.

\section*{Acknowledgements}
The work of the authors was supported by
a Melbourne Postgraduate Research Award and the Australian Research Council 
respectively. 

\appendix

\section{Appendix}
\begin{proposition}\label{p5.1}
Let $w(x)$ be a weight function with support on successive integers $a,a+1,\dots,b$ ($a>b$),
and suppose $w(x)$ is even about the midpoint $m:=(a+b)/2$ so that 
\begin{equation}\label{s1}
w(x-m) = w(-(x-m)), \qquad \pm (x-m) \in \{a,a+1,\dots,b\}
\end{equation}
Let $\{p_n(x)\}_{n=0,1,\dots}$ be the family of monic orthogonal
polynomials $p_n(x)$ of degree $n$, with the orthogonality relationship
\begin{equation}\label{orthogrel}
\sum_{x=a}^b w(x) p_n(x) p_m(x) \ = (p_n,p_n) \delta_{m,n}
\end{equation}
By the property (\ref{s1}), $p_j(x-m) = (-1)^j p_j(-(x-m))$, so that $p_j(x)$ is even about $m$ 
for $j$ even, and odd about $m$ for $j$ odd.
Let $w(x)$ be as in (\ref{s1}) and $\{p_n(x)\}_{n=0,1,\dots}$ be as in (\ref{orthogrel}). One has
\begin{eqnarray}
\sum_{x_1=a}^b \dots \sum_{x_n=a}^b  w(x_1) \dots w(x_n) (\Delta(x-m)^2)^2 &\!\! \! = \!\!\! & n! \prod_{j=0}^{n-1} (p_{2j}, p_{2j}) \label{orthogodd}\\
\label{orthogeven} \sum_{x_1=a}^b \dots \sum_{x_n=a}^b  w(x_1) \dots w(x_n)( \Delta(x-m)^2)^2 \prod_{i=1}^n (x_i -m)^2 & \!\!\! = \!\!\! & n! \prod_{j=0}^{n-1} (p_{2j+1}, p_{2j+1})\end{eqnarray}
\end{proposition}

\begin{proof}
Consider first (\ref{orthogodd}). According to the Vandermonde determinant identity
\begin{equation}\label{dd}
\Delta((x-m)^2) = \det [ (x_i - m)^{2(j-1)} ]_{i,j=1,\dots,n} = \det [p_{2(j-1)}(x_i) ]_{i,j=1,\dots,n}
\end{equation}
where to obtain the second equality the fact that $p_{2j}(x)$ is even about $x=m$ has been used
(since $\Delta((x-m)^2)$ is unchanged (up to sign) by $x_i \leftrightarrow x_j$ and by
$(x_i-m) \mapsto -(x_i-m)$ this is referred to as a type $B$ identity).
Thus the LHS of (\ref{orthogodd}) can be rewritten
\begin{equation}\label{s2}
 \sum_{x_1=a}^b \dots \sum_{x_n=a}^b  w(x_1) \dots w(x_n) \det[p_{2(j-1)}(x_i)]\det[p_{2(j-1)}(x_i)] 
\end{equation}
Since both determinants are antisymmetric in $\{x_j\}$, while the remaining factors in the summand
are symmetric, we can replace one of the determinants by $n!$ times its diagonal term
$\prod_{i=1}^n p_{2(i-1)}(x_i)$. We can then perform each sum column-by-column in the remaining
determinant to reduce (\ref{s2}) to
$$
n! \det \Big [ \sum_{x=a}^b w(x) p_{2(i-1)}(x) p_{2(j-1)}(x) \Big ]_{i,j=1,\dots,n}
$$
Now making use of (\ref{orthogrel}), (\ref{orthogodd}) results.

Regarding (\ref{orthogodd}), instead of (\ref{dd}) we use
$$
\Delta((x-m)^2) \prod_{i=1}^n (x_i -m) =  \det [p_{2j-1}(x_i) ]_{i,j=1,\dots,n}
$$
which is a consequence of the Vandermonde determinant identity and the fact that
$p_{2j-1}(x)$ is odd about $x=m$, and proceed similarly. 
\end{proof}

The special case $p={1 \over 2}$ of the monic Krawtchouk polynomials obey the discrete orthogonality equation
\begin{equation}
\sum_{x=0}^{N} \frac{1}{2^N x! (N-x)!} p_{a,N}(x) p_{b,N}(x) = \frac{a!}{2^{2a} (N-a)!}\delta_{a,b}
\end{equation}
(see e.g.~\cite{KK96}).
Setting $N = 2M+1$, $x= t+M$ this reads
\begin{equation}
\sum_{t=-M}^{M+1} \frac{1}{2^{2M+1} (t+M)! (M+1-t)!} p_{a,2M+1}(t+M) p_{b,2M+1}(t+M) = \frac{a!}{2^{2a} (2M+1-a)!} \delta_{a,b}
\end{equation}
If we define a new family of monic polynomials, $q_{n,N}$, by $q_{n,N}(x) = p_{n,N}(N+\frac{1}{2} +x)$ then we have
\begin{equation}\label{89}
\sum_{t=-M}^{M+1} \frac{1}{2^{2M+1}(t+M)!(M+1-t)!}q_{a,2M+1}\left(t-\frac{1}{2}\right) q_{b,2M+1}\left(t-\frac{1}{2}\right) =  \frac{a!}{2^{2a} (2M+1-a)!} \delta_{a,b}
\end{equation}

The midpoint of the support of the weight in (\ref{89}) is $t=1/2$, and furthermore the weight is
symmetrical about this point. Hence we can apply Proposition \ref{p5.1} to deduce that
\begin{eqnarray*}
&& \sum_{x_1=1}^{M+1} \dots \sum_{x_n=1}^{M+1} \prod_{i=1}^n \frac{1}{(M+x_i)!(M+1-x_i)!} \Delta\left(( x-\frac{1}{2})^2\right)^2 = n! C_{2n-1}
\\
&&\sum_{x_1=1}^{M+1} \dots \sum_{x_n=1}^{M+1} \prod_{i=1}^n \frac{(x_i-\frac{1}{2})^2}{(M+x_i)!(M+1-x_i)!} \Delta\left(( x-\frac{1}{2})^2\right)^2 = n! {C_{2n}}
\end{eqnarray*}
where
\begin{eqnarray}
C_{2n-1} & = & 2^{2Mn} \prod_{j=0}^{n-1} \frac{(2j)!}{2^{4j}(2(M-j)+1)!}  \label{92} \\
C_{2n} & = & 2^{2Mn} \prod_{j=0}^{n-1} \frac{(2j+1)!}{2^{4j+2}(2M-2j)!}  \label{94}
\end{eqnarray}
Here the range of summation has been halved by noting $\sum_{x=-M}^{M+1} =
2 \sum_{x=0}^{M+1}$ when the summand is even about $x=1/2$. Minor manipulation of
(\ref{92}) and (\ref{94}) reclaims the normalizations in \eqref{67} and  \eqref{even1line}.


\begin{thebibliography}{10}

\bibitem{Ba01a}
Y.~Baryshnikov, \emph{{GUE}s and queues}, Probab. Theory Relat. Fields
  \textbf{119} (2001), 256--274.

\bibitem{BF08}
A.~Borodin and P.~Ferrari, \emph{Anisotropic growth of random surfaces in $2+1$
  dimensions}, arXiv:0804.3035, 2008.

\bibitem{BFPS06}
A.~Borodin, P.L. Ferrari, M.~Pr\"ahoffer, and T.~Sasamoto, \emph{Fluctuation
  properties of the {TASEP} with periodic initial configuration}, J. Stat.
  Phys. \textbf{129} (2006), 1055--1080.

\bibitem{BP07}
A.~Borodin and S.~P\'ech\'e, \emph{Airy kernel with two sets of parameters in
  directed percolation and random matrix theory}, J. Stat. Phys. \textbf{132}
  (2008), 275--290.

\bibitem{De08a}
M.~Defosseux, \emph{Orbit measures and interlaced determinantal point
  processes}, Compte Rendus Math. \textbf{346} (2008), 783--788.

\bibitem{EKLP92}
N.~Elkies, G.~Kuperberg, M.~Larsen, and J.~Propp, \emph{Alternating sign
  matrices and domino tilings {I}}, J. Algebraic Combin. \textbf{1} (1992),
  111--132.

\bibitem{EF05}
S.-P. Eu and T-S. Fu, \emph{A simple proof of the {A}ztec diamond theorem},
  Elec. J. Comb. \textbf{12} (2005), \#R18.

\bibitem{Fo10}
P.J. Forrester, \emph{Log-gases and random matrices}, Princeton University
  Press, Princeton, NJ, 2010.

\bibitem{FN08}
P.J. Forrester and T.~Nagao, \emph{Determinantal correlations for classical
  projection processes}, arXiv:0801.0100, 2008.

\bibitem{FN08a}
P.J. Forrester and E.~Nordenstam, \emph{The anti-symmetric {GUE} minor
  process}, Moscow Math. J. \textbf{9} (2008), 749--774.

\bibitem{FR02}
P.J. Forrester and E.M. Rains, \emph{Correlations for superpositions and
  decimations of {L}aguerre and {J}acobi orthogonal matrix ensembles with a
  parameter}, Prob. Theory Related Fields \textbf{130} (2004), 518--576.

\bibitem{JPS98}
W.~Jockush, J.~Propp, and P.~Shor, \emph{Random domino tilings and the arctic
  circle theorem}, math.CO/9801068, 1998.

\bibitem{Jo01x}
K.~Johansson, \emph{Discrete orthogonal polynomial ensembles and the
  {P}lancherel measure}, Ann. Math. \textbf{153} (2001), 259--296.

\bibitem{Jo02}
\bysame, \emph{Non-intersecting paths, random tilings and random matrices},
  Prob. Theory Related Fields \textbf{123} (2002), 225--280.

\bibitem{Jo05a}
\bysame, \emph{The arctic circle boundary and the {A}iry process}, Ann. Probab.
  \textbf{33} (2005), 1--30.

\bibitem{JN06}
K.~Johansson and E.~Nordenstam, \emph{Eigenvalues of {GUE} minors}, Elect. J.
  Probability \textbf{11} (2006), 1342--1371.

\bibitem{KK96}
R.~Koekoek and R.F. Swarttouw, \emph{The {A}skey-scheme of hypergeometric
  orthogonal polynomials and its $q$-analogue}, arXiv:math/9602214, 1996.

\bibitem{MOW09}
A.P. Metcalfe, N.~O'Connell, and J.~Warrren, \emph{Interlaced processes on the
  circle}, Ann. Inst. Poincar\'e Probab. Statist. \textbf{45} (2009),
  1165--1184.

\bibitem{No09}
E.J.G. Nordenstam, \emph{On the shuffling algorithm for domino tilings},
  Electronic J. Prob. \textbf{15} (2009), 75--95.

\bibitem{PS99a}
J.~Propp and R.~Stanley, \emph{Domino tilings with barriers}, J. Comb. Th.
  Series A \textbf{87} (1999), 347--356.

\bibitem{Ra00}
E.M. Rains, \emph{Correlations for symmetrized increasing subsequences},
  math.CO/0006097, 2000.

\bibitem{Sa05}
T.~Sasamoto, \emph{Spatial correlations of the {1D} {KPZ} surface on a flat
  substrate}, J. Phys. A \textbf{38} (2005), L549--L556.

\end{thebibliography}

\providecommand{\bysame}{\leavevmode\hbox to3em{\hrulefill}\thinspace}
\providecommand{\MR}{\relax\ifhmode\unskip\space\fi MR }
\providecommand{\MRhref}[2]{%
  \href{http://www.ams.org/mathscinet-getitem?mr=#1}{#2}
}
\providecommand{\href}[2]{#2}

\end{document}